\documentclass[aps,pra,twocolumn,english,superscriptaddress,floatfix,longbibliography]{revtex4-2}

\usepackage[english]{babel}
\usepackage{textcomp,xcolor,natbib}

\usepackage{amsmath,amssymb,graphicx,bbm,hyperref,latexsym,theorem}
\usepackage[normalem]{ulem}
\newcommand{\bignone}{}
\newcommand{\mathd}{\mathrm{d}}
\newcommand{\nobracket}{}
\newcommand{\tmem}[1]{{\em #1\/}}
\newcommand{\tmmathbf}[1]{\ensuremath{\boldsymbol{#1}}}
\newcommand{\tmop}[1]{\ensuremath{\operatorname{#1}}}
\newcommand{\tmrsub}[1]{\ensuremath{_{\textrm{#1}}}}
\newcommand{\tmtextbf}[1]{\text{{\bfseries{#1}}}}
\newcommand{\tmtexttt}[1]{\text{{\ttfamily{#1}}}}
\newenvironment{proof}{\noindent\textbf{Proof\ }}{\hspace*{\fill}$\Box$\medskip}
\newcounter{nnacknowledgments}

{\theorembodyfont{\rmfamily}\newtheorem{acknowledgments*}[nnacknowledgments]{Acknowledgments}}
\newtheorem{definition}{Definition}
\newtheorem{lemma}{Lemma}
\newtheorem{theorem}{Theorem}
\usepackage[utf8]{inputenc}
\usepackage[T1]{fontenc}

\begin{document}

\title{Classical pair of states as optimal pair for quantum distinguishability quantifiers}

\author{Bassano Vacchini}
\email{bassano.vacchini@mi.infn.it}
\affiliation{Dipartimento di Fisica Aldo Pontremoli, Università degli Studi di Milano, via Celoria 16, 20133 Milan, Italy}
\affiliation{Istituto Nazionale di Fisica Nucleare, Sezione di Milano, via Celoria 16, 20133 Milan, Italy}

\author{Andrea Smirne}
\affiliation{Dipartimento di Fisica Aldo Pontremoli, Università degli Studi di Milano, via Celoria 16, 20133 Milan, Italy}
\affiliation{Istituto Nazionale di Fisica Nucleare, Sezione di Milano, via Celoria 16, 20133 Milan, Italy}

\author{Nina Megier}
\affiliation{Dipartimento di Fisica Aldo Pontremoli, Università degli Studi di Milano, via Celoria 16, 20133 Milan, Italy}
\affiliation{Istituto Nazionale di Fisica Nucleare, Sezione di Milano, via Celoria 16, 20133 Milan, Italy}
\affiliation{International Centre for Theory of Quantum Technologies, University of Gdańsk, 80-308 Gdańsk, Poland}

\begin{abstract}
  The capability to quantitatively distinguish quantum states is of great
  importance for a variety of tasks, and has recently played an important role
  in the study of quantum reduced dynamics and their characterization in terms
  of memory effects. A crucial property of quantum distinguishability
  quantifiers considered in the latter framework is the contractivity under
  the action of completely positive trace preserving maps. We show that this
  requirement warrants that the pairs on which these quantifiers attain their
  maximal value are pairs of orthogonal, and in this sense classical, states.
\end{abstract}

\maketitle

\section{Introduction}

Quantifying the difference between probability distributions is essential in
many and diverse fields, notably including statistics and information
theory. In particular, $f$-divergences provide a unified framework for
defining several distinguishability measures, such as the Kullback--Leibler
and the Jensen-Shannon divergences, that further possess well-defined
operational interpretations.

In the quantum setting, distinguishing probability distributions extends to
the problem of distinguishing quantum states. Accordingly, $f$-divergences for
probability distributions can be generalized to define quantum
distinguishability quantifiers, including the quantum relative entropy, the
Holevo quantity for two-element ensembles, and the trace distance. These
quantifiers are widely employed in quantum information, computation and
communication, as well as in investigating general dynamical evolutions of
quantum systems, in particular to identify the presence of memory effects or
nonclassicality {\cite{Bengtsson2017,Chruscinski2022a,Vacchini2024}}.

A fundamental property of quantum distinguishability quantifiers is their
contractivity under completely positive trace preserving maps.
The latter represents a broad family of quantum-state transformations -- arising,
for example, from non-selective measurements or interactions with the initially factorized
environment -- and the contractivity ensures that such processes cannot enhance the
capability to distinguish among the states of the system.

In this article, we establish a direct connection between the contractivity of
quantum distinguishability quantifiers under completely positive trace
preserving maps and their optimal pairs, i.e., the pair of quantum states for
which the quantifier attains its maximum value
{\cite{Wissmann2012a,Breuer2012a}}. Specifically, we demonstrate that if a
quantum distinguishability quantifier is contractive under completely positive
maps, and its invariance guarantees invertibility of the map when restricted
to commuting states, then the quantifier reaches its maximum value only and
for all orthogonal state pairs, i.e., states that share a common eigenbasis
and can thus be considered in this respect a classical pair. The relevance of
classical states in assessing properties of quantum distinguishability
quantifiers has been recently put into evidence also in
{\cite{Osan2022a,Bussandri2023a,Lanier2025a}}.

Moreover, we present in a unified framework well-known contractive quantum
distinguishability quantifiers for which the optimal pairs are all orthogonal
states, as well as non-contractive quantifiers for which orthogonality does
not guarantee optimality. Indeed, such a different behavior is possible and even
desirable for certain information tasks {\cite{DePalma2021a}}. We further
investigate the contractivity of the quantum distinguishability quantifiers by
relating it to invariance properties of the quantifiers. Specifically, we
show that if the quantifier is contractive under the partial trace operation
or jointly convex in its arguments, contractivity under completely positive
trace preserving maps is in one-to-one correspondence with the invariance
under unitary transformations and assignment maps.

The rest of the paper is organized as follows. In Sect. \ref{sect:cq} we
recall the definition and key properties of classical $f$-divergences, also
presenting some examples; we then move to the quantum setting, introducing the
notion of quantum distinguishability quantifiers, the properties they satisfy
and some significant examples that reduce to classical $f$-divergences when
evaluated on couples of commuting states. In Sect. \ref{sect:mp}, we derive
the main result of the paper, i.e., the theorem proving that orthogonality of
a couple of quantum states is a necessary and sufficient condition for their
optimality for contractive distinguishability quantifiers, and we show how
some well-known quantifiers can be understood within a unified framework by
virtue of the theorem; instead, in Sect. \ref{sect:cex}, we present quantum
distinguishability quantifiers for which contractivity under completely
positive trace preserving maps does not hold, orthogonality being a necessary
but not sufficient condition for the optimal pair. In Sect. \ref{sect:cai},
the relationships between contractivity and invariance properties for a
quantum distinguishability quantifier are presented. Finally, concluding
remarks are given in Sect. \ref{sect:con}.

\section{Classical and quantum distinguishability quantifiers}\label{sect:cq}

We introduce quantum distinguishability quantifiers starting from the
classical framework, which allows us to take the unified perspective of
$f$-divergences without lack of rigour but avoiding to deal with the many
subtle aspects involved in the definition of quantum $f$-divergences. Indeed,
due to non commutativity of the operators, different quantum generalizations
of $f$-divergences can be considered, see e.g. {\cite{Hiai2011a}}. It
is important to stress that, while $f$-divergences include most of the
well-known and studied distinguishability quantifiers {\cite{Crooks}}, not all
of them can be recast in this form. We mention in this respect the
Bhattacharyya coefficient or classical fidelity
{\cite{Bhattacharyya1943a,Bhattacharyya1946a}} and correspondingly the
quantum fidelity {\cite{Jozsa1994a}} based on Uhlmann's transition probability
{\cite{Uhlmann1976a}}. The emphasis in this concise presentation is on the
role played by the contractivity of the considered distinguishability
quantifiers with respect to the action of transformations that can describe
the effect of a noisy transmission channel, be it classical or quantum, also called data processing inequality. This
feature leads to important invariances for the distinguishability quantifiers,
in both classical and quantum formulations. We will consider in particular distinguishability
quantifiers for which saturation of the data processing inequality under the action of a noisy transformation
on a collection of states implies reversibility of the transformation on these states, as it happens 
for $f$-divergences  in the classical case when the function $f$ is strictly convex. This
result will prove instrumental to the characterization of pairs on which these
contractive distinguishability quantifiers take their maximum value.

\subsection{Distinguishability quantifiers expressible as
\texorpdfstring{$f$-divergences}{f-divergences}}\label{sect:fdiv}

The motivation for the introduction of $f$-divergences is the comparison of
classical probability distributions, typically describing signals arising from
two distinct sources, to investigate the effect of the unavoidable noise on
the propagation of these signals. The guiding idea is that the action of any
realistic noise source along the transmission line should reduce the extent to
which the signals coming from the two distinct sources can be seen to differ
or diverge from one another. It is therefore of interest to introduce ways to
compare probability distributions that keep this crucial feature into account.
This task is accomplished by $f$-divergences, introduced in
{\cite{Csiszar1963a,Csiszar1967a}} and independently in {\cite{Ali1966a}},
that are a class of distinguishability quantifiers between probability
distributions specified by a suitable convex function $f$.

Let $f$ be a non-negative continuous convex function defined on the positive real line, and
suppose for the sake of simplicity to have a discrete set $X$ of outcomes for
the probability distributions.

\begin{definition}
  (Classical $f$-divergence {\cite{Csiszar1963a,Ali1966a}})
  
  Given the distributions $\tmmathbf{p}$ with elements $\{ p_i \}_{i \in X}$
  and $\tmmathbf{q}$ with elements $\{ q_i \}_{i \in X}$ the quantity
  \begin{eqnarray}
    D_f (\tmmathbf{p}, \tmmathbf{q}) & = & \sum_{i \in X} q_i f \left(
    \frac{p_i}{q_i} \right) \bignone,  \label{eq:f}
  \end{eqnarray}
  set equal to zero for vanishing $q_i$, is the $f$-divergence between
  $\tmmathbf{p}$ and $\tmmathbf{q}$.
\end{definition}

This expression taking on positive values quantifies how $\{ p_i \}_{i \in X}$
differs from $\{ q_i \}_{i \in X}$. This quantifier is generally neither
symmetric nor bounded, which is another motivation for the name divergence,
and typically does not obey the triangle inequality. It has however a number
of interesting properties, that we detail below together with a simple proof,
putting into evidence those relevant for the subsequent treatment. As we shall
stress, it is often the case that skewed versions of the original
$f$-divergence, in which the comparison is between one probability
distribution $\{ p_i \}_{i \in X}$ and a proper mixture of $\{ p_i \}_{i \in
X}$ and $\{ q_i \}_{i \in X}$, still fall within the class for a different
choice of convex function $f$. These skewed divergences retain all relevant
properties, but turn out to be bounded, so that they can be normalized and
made symmetric.

\subsection{General properties of \texorpdfstring{$f$-divergences}{f-divergences}}\label{sect:inv}

Independently on the choice of convex function $f$, the quantity $D_f$ has a
number of relevant properties, that we introduce below.

\textit{i.) $D_f$ is non-increasing under the action of a stochastic map}

Stochastic maps correspond to the most general transformations between
probability distributions, generally describing the effect of a noisy
transmission line. Let $T$ be a stochastic map to be identified with the
matrix $T_{ij}$, with positive entries satisfying $\sum_{i \in X} \bignone
T_{ij} = 1$ for all $j \in X$, acting as follows on the considered probability
distributions
\begin{eqnarray}
  (T\tmmathbf{p})_i & = & \sum_{j \in X} \bignone T_{ij} p_j . 
\end{eqnarray}
We then have
\begin{eqnarray}
  D_f (T\tmmathbf{p}, T\tmmathbf{q}) & = & \sum_{i \in X} (T\tmmathbf{q})_i f
  \left( \frac{(T\tmmathbf{p})_i}{(T\tmmathbf{q})_i} \right) \bignone
  \nonumber\\
  & = & \sum_{i \in X} (T\tmmathbf{q})_i f \left( \sum_{j \in X} \frac{T_{ij}
  q_j}{(T\tmmathbf{q})_i} \frac{p_j}{q_j} \right) \nonumber\\
  & \leqslant & \sum_{i, j \in X} T_{ij} q_j f \left( \frac{p_j}{q_j} \right)
  \nonumber\\
  & = & \sum_{j \in X} q_j f \left( \frac{p_j}{q_j} \right), 
\end{eqnarray}
where the fact that $\{ T_{ij} q_j / (T\tmmathbf{q})_i \}_{j \in X}$ is a
probability distribution for any $i$ has allowed to exploit convexity of $f$
and in the last line we have used the fact that $T$ is a stochastic map. We
thus have the crucial result
\begin{eqnarray}
  D_f (T\tmmathbf{p}, T\tmmathbf{q}) &
                                       \leqslant & D_f (\tmmathbf{p},
  \tmmathbf{q}),  \label{eq:contraiC}
\end{eqnarray}
that is the contractivity property with respect to the action of any
stochastic map. This inequality can be seen as a formulation of the data
processing inequality {\cite{Nielsen2000}}.

\textit{ii.) $D_f$ is invariant under the action of invertible stochastic
maps}

If the map $T$ is invertible, in the sense that it admits an inverse within
the set of stochastic maps, we have, exploiting Eq.~(\ref{eq:contraiC}) twice,
\begin{eqnarray}
  D_f (\tmmathbf{p}, \tmmathbf{q}) & = & D_f (\{ (T^{- 1} T\tmmathbf{p})_i \},
  \{ (T^{- 1} T\tmmathbf{q})_i \}) \nonumber\\
  & \leqslant & D_f (\{ (T\tmmathbf{p})_i \}, \{ (T\tmmathbf{q})_i \})
  \nonumber\\
  & \leqslant & D_f (\tmmathbf{p}, \tmmathbf{q})  \label{eq:TU}
\end{eqnarray}
and therefore invariance
\begin{eqnarray}
  D_f (T\tmmathbf{p}, T\tmmathbf{q}) & = & D_f (\tmmathbf{p}, \tmmathbf{q}) . 
  \label{eq:equalC}
\end{eqnarray}
As we shall discuss in Sect.~\ref{sect:suff}, if the function $f$ is strictly
convex, this equality implies invertibility of $T$ on the given pair
of probability distributions {\cite{Hiai2011a,Hiai2017a}}.

\textit{iii.) $D_f$ is invariant under extension of the distributions to bivariate distributions with the same marginal}

Let us now embed our probability distributions in factorized bivariate
distributions on the Cartesian product $X \times Y$, with a common marginal,
that is $\{ P_{ij} \}_{i, j \in X \times Y}$ and $\{ Q_{ij} \}_{i, j \in X
\times Y}$ such that $P_{ij} = p_i w_j$ and $Q_{ij} = q_i w_j$ with the same
probability distribution $\{ w_i \}_{i \in Y}$. We then have
\begin{eqnarray}
  D_f (\tmmathbf{Q}, \tmmathbf{P}) & = & \sum_{i, j \in X \times Y} q_i w_j f
  \left( \frac{p_i}{q_i} \right) \bignone, \nonumber\\
  & = & D_f (\tmmathbf{p}, \tmmathbf{q}),  \label{eq:tensor}
\end{eqnarray}
a simple property that will become important when exploiting the same
conceptual framework in quantum theory, where the description of composite
systems unveils new features. It corresponds to the fact that no disturbance
arises if our signal is transmitted together with other signals provided no
correlations are established. Note that in Eq.~(\ref{eq:tensor}) we have kept
the same notation for the $f$-divergence acting on bivariate distributions.

\textit{iv.) $D_f$ is jointly convex in its arguments}

We finally consider joint convexity, expressed by
\begin{eqnarray}
  D_f \left( \sum_{k = 1}^n \mu_k \bignone \tmmathbf{q}_k, \sum_{k = 1}^n
  \mu_k \bignone \tmmathbf{p}_k \right) & \leqslant & \sum_{k = 1}^n \mu_k D_f
  \left( \bignone \tmmathbf{q}_k, \tmmathbf{p}_k \right) .  \label{eq:jointc}
\end{eqnarray}
Given the expression Eq.~(\ref{eq:f}) of the $f$-divergence, to prove this
property it is enough to show joint convexity of the function
\begin{eqnarray}
  g (x, y) & = & xf \left( \frac{y}{x} \right) . 
\end{eqnarray}
Convexity of $f$ warrants that the Hessian of the function $g$, that is the
square matrix of the partial second derivatives of $g$, is always positive
semi-definite. This property is equivalent to convexity of the function $g$.

\subsection{Relevant examples of classical \texorpdfstring{$f$-divergences}{f-divergences}}\label{sect:fex}

\subsubsection*{Kullback-Leibler divergence}The most prominent $f$-divergence
is obtained considering the convex function
\begin{eqnarray}
  f_{\tmop{kl}} (t) & = & t \log t  \label{eq:fkl}
\end{eqnarray}
leading to the famous Kullback-Leibler divergence or classical relative
entropy, namely
\begin{eqnarray}
  D_{f_{\tmop{kl}}} (\tmmathbf{p}, \tmmathbf{q}) & = & \sum_{i \in X} p_i \log
  \left( \frac{p_i}{q_i} \right) . 
\end{eqnarray}
In this case, as well-known, the triangle inequality fails and the
quantity can be infinite. As we shall see this quantity has a natural quantum
counterpart, for which skewed variants have been considered, that allow to
avoid infinite values and to introduce triangle-like inequalities to bound the
variation of the quantity when varying one of the arguments
{\cite{Audenaert2014a,Audenaert2014c,Smirne2022b}}.

\subsubsection*{Skewed relative entropies}These skewed relative entropies can
be introduced also in the classical framework and for the case of the convex
functions
\begin{eqnarray}
  f^{\mu}_{\tmop{skew}} (t) & = & \frac{1}{\log (1 / \mu)} t \log \left(
  \frac{t}{\mu t + (1 - \mu)} \right),  \label{eq:fcsd}
\end{eqnarray}
parametrized by a skewing parameter $\mu \in (0, 1)$, lead to the classical
skew divergences
\begin{equation}
  D_{f^{\mu}_{\tmop{skew}}} (\tmmathbf{p}, \tmmathbf{q}) =  \frac{1}{\log
  (1 / \mu)} \sum_{i \in X} p_i \log \left( \frac{p_i}{\mu p_i + (1 - \mu)
  q_i} \right). 
\end{equation}
On similar grounds for the convex functions
\begin{equation}
  f^{\mu}_{\tmop{hsd}} (t)  =  \mu t \log t - (\mu t + (1 - \mu)) \log ((1 -
  \mu) + \mu t)  \label{eq:fhsd}
\end{equation}
we obtain the divergences
\begin{eqnarray}
  D_{f^{\mu}_{\tmop{hsd}}} (\tmmathbf{p}, \tmmathbf{q}) & = & \mu \sum_{i \in
  X} p_i \log p_i + (1 - \mu) \sum_{i \in X} q_i \log q_i \\
  &  & - \sum_{i \in X} (\mu p_i + (1 - \mu) q_i) \log (\mu p_i + (1 - \mu)
  q_i), \nonumber
\end{eqnarray}
corresponding to a special case of the Holevo quantity {\cite{Nielsen2000}}.

\subsubsection*{Kolmogorov distance}Another example of $f$-divergence that
deserves being singled out in the present context is obtained for
\begin{eqnarray}
  f_{\tmop{vd}} (t) & = & \frac{1}{2} | 1 - t |,  \label{eq:fqsd}
\end{eqnarray}
leading to the variation or Kolmogorov distance {\cite{Kolmogorov1963a}}
\begin{eqnarray}
  D_{f _{\tmop{vd}}} (\tmmathbf{p}, \tmmathbf{q}) & = & \frac{1}{2} \sum_{i
  \in X} | p_i - q_i | . 
\end{eqnarray}
In this case the obtained divergence becomes a distance, thus warranting the
triangle inequality. As discussed in {\cite{Csiszar1975a}}, this is generally
not the case, even though divergences can have the properties of squared
distances. Making reference to the considered examples, taking the value of
the skewing parameter $\mu = \frac{1}{2}$, both $D_{f^{\mu}_{\tmop{skew}}}$
and $D_{f^{\mu}_{\tmop{hsd}}}$ lead to the same expression, known as
Jensen-Shannon divergence
\begin{eqnarray}
  D_{f_{\tmop{js}}} (\tmmathbf{p}, \tmmathbf{q}) & = & \frac{1}{2} \sum_{i \in
  X} p_i \log \left( \frac{p_i}{\tfrac{1}{2} (p_i + q_i)} \right) \nonumber\\
  &&
  + \frac{1}{2} \sum_{i \in X} q_i \log \left( \frac{q_i}{\tfrac{1}{2} (p_i +
  q_i)} \right), 
\end{eqnarray}
whose square root has indeed been shown to be a distance satisfying the
triangle inequality {\cite{Endres2003a}}, as it happens for the
Hellinger divergence.

\subsection{Sufficiency and data processing inequality}\label{sect:suff}

An important notion in the framework of distinguishability quantifiers is the
notion of sufficiency or reversibility of a stochastic transformation with
respect to a set of probability distributions {\cite{Hiai2011a,Jencova2012a}}.
\begin{definition}
  (Sufficiency of stochastic maps)
  Consider a set $\wp$ of probability distributions and a stochastic map $T$.
If $T$ is invertible on any  $\tmmathbf{p} \in\wp$, it is said
  to be sufficient or reversible with respect to $\wp$.
\end{definition}
We recall that $T$ is invertible on $\tmmathbf{p}$ if
there exists a stochastic map $T^{- 1}$ such that
\begin{equation}
    (T^{- 1} \circ T) \tmmathbf{p} = \tmmathbf{p}.  \label{eq:invariap}  
\end{equation}
The action of the transformation $T$ does not lead to any loss of information
that would reduce the capability to discriminate between $\tmmathbf{q},\tmmathbf{p}\in\wp$, so that the statistics at the output is sufficient for this
task. Indeed, if we now consider distinguishability quantifiers given by
$f$-divergences, as already shown in Sect.~\ref{sect:inv} invertibility of the
stochastic map ensures saturation of the data processing inequality. On the
other hand if the function $f$ is strictly convex, validity of
\begin{equation}
    D_{f}(T\tmmathbf{p} ,T\tmmathbf{q} ) =  D_{f}(\tmmathbf{p},\tmmathbf{q} ) \label{eq:invariaFIN}
 \end{equation}
 implies that $T$ is invertible on the
pair $\tmmathbf{q}$ and $\tmmathbf{p}$, that is there exists a stochastic map
$T^{- 1}$ such that Eqs.~(\ref{eq:invariap}) holds for both $\tmmathbf{p}$ and $\tmmathbf{q}$
\cite{Hiai2011a,Jencova2012a}. In this case invariance of the considered $f$-divergence with respect to the action of the map on a collection of probability distributions implies its sufficiency with respect to these probability distributions:
\begin{equation}
D_{f}(T\tmmathbf{p} ,T\tmmathbf{q} ) =  D_{f}(\tmmathbf{p},\tmmathbf{q} ) \Rightarrow 
T \text{ sufficient on }\tmmathbf{p}\text{ and }\tmmathbf{q}
\label{eq:inv_suff}
\end{equation}
for $f$ strictly convex.

\subsection{Quantum statistical description}\label{sect:qm}

As it is well-known, quantum mechanics can be seen as a probability theory,
with a distinct formal structure with respect to the classical one, recovering
a classical probabilistic description when features related to the
non-commutativity of the spaces of states and operators are not relevant
{\cite{Holevo2001,Vacchini2024}}. In the framework of quantum theory the state
of a system is given by a statistical operator rather than a probability
distribution, that is an element of the convex set $\mathcal{S (\mathcal{H})}$
of positive and normalized trace class operators on the Hilbert space of the
system $\mathcal{H}$, and the natural transformations between states are given
by completely positive trace preserving maps. In analogy to the classical
case, we are then led to consider quantum distinguishability quantifiers that
provide an estimate of how a given statistical operator $\rho$ differs or
diverges from another statistical operator $\sigma$, according to
$\mathfrak{S} (\rho, \sigma) \in [0, + \infty)$, where we have included the
situation in which this quantifier can be unbounded. We are thus led to
consider the following definition.

\begin{definition}
  (Quantum distinguishability quantifier)
  
  Consider an arbitrary pair of quantum states $\rho, \sigma \in \mathcal{S
  (\mathcal{H})}$. A quantum distinguishability quantifier is a non-negative
  function $\mathfrak{S} (\rho , \sigma)$ of $\rho$ and $\sigma$ that
  satisfies contractivity under the action of an arbitrary completely positive
  trace preserving map $\Phi$ according to
  \begin{eqnarray}
    \mathfrak{S} (\Phi [\rho], \Phi [\sigma]) & \leqslant & \mathfrak{S} (\rho
    , \sigma) .  \label{eq:inv1}
  \end{eqnarray}
\end{definition}

As in the classical case, the defining property for such distinguishability
quantifiers is the contractivity under the action of transformations that
describe in the quantum case general state changes and therefore all possible
disturbances, namely completely positive trace preserving maps.

Along the same line of proof as Eq.~(\ref{eq:TU}), this property implies
invariance under the action of an invertible transformation, that is 
\begin{equation}
      \mathfrak{S} (\Phi [\rho], \Phi [\sigma]) = \mathfrak{S} (\rho
    , \sigma)   \label{eq:sussq}
\end{equation}
if $\Phi$ is invertible on the pair $\rho,\sigma$ with an inverse given by a completely positive trace preserving map.   
In particular, we have
\begin{eqnarray}
  \mathfrak{S} (U \rho U^{\dag}, U \sigma U^{\dag}) & = & \mathfrak{S} (\rho ,
  \sigma),  \label{eq:invU}
\end{eqnarray}
with $U$ a unitary operator and no constraint on $\rho$ and $\sigma$, given that unitary transformations
\begin{eqnarray}
  \mathcal{U} [\rho] & = & U \rho U^{\dag}  \label{eq:U}
\end{eqnarray}
are the only completely positive maps invertible on the whole set of states.

Another important consequence of the contractivity property
Eq.~(\ref{eq:inv1}) is stability of the distinguishability quantifier
$\mathfrak{S}$ with respect to the extension of its entries to a tensor
product with the same marginal, that is
\begin{eqnarray}
  \mathfrak{S} (\rho \otimes \tau, \sigma \otimes \tau) & = & \mathfrak{S}
  (\rho , \sigma),  \label{eq:invTensor}
\end{eqnarray}
where the tensor product is taken with a statistical operator $\tau$ in an
arbitrary auxiliary space, corresponding to so-called assignment maps
\begin{eqnarray}
  \mathcal{A}_{\tau} [\rho] & = & \rho \otimes \tau .  \label{eq:A}
\end{eqnarray}
Note that also in this case, as in Eq.~(\ref{eq:tensor}), according to
standard usage we have kept the same notation for the distinguishability
quantifier independently of the space on which the pair of states to be
compared acts. This property is the quantum counterpart of
Eq.~(\ref{eq:tensor}) and describes the invariance of the comparison of the
two states with respect to the involvement of other degrees of freedom,
provided they are uncorrelated from the system of interest. Its proof relies
on the fact that both the partial trace and the assignment map are completely
positive trace preserving maps, and the partial trace acts as left inverse of
the assignment map.

The non-commutativity of quantum states, that correspond to classical
probability distributions, and of quantum observables, that correspond to
random variables, makes non-trivial the problem of extending the approach of
Sect.~\ref{sect:fdiv} to the quantum framework. The definition of a proper
quantum distinguishability quantifier starting from a convex function $f$,
which provided a useful and unified perspective in the classical framework,
involves many subtleties. Indeed, to a given function $f$ distinct
distinguishability quantifiers can be associated {\cite{Hiai2011a}}. Moreover,
even if $f$ is a strictly convex function providing an operator-convex
function, the preservation of the corresponding $f$-divergence under the
action of a given completely positive trace preserving map does not imply
reversibility of the quantum transformation {\cite{Jencova2012a}}. We thus
introduce the following definition.

\begin{definition}
  \label{def:qfd}(Quantum $f$-divergence)
  
  Consider an arbitrary pair of quantum states $\rho, \sigma \in \mathcal{S
  (\mathcal{H})}$. A quantum $f$-divergence $\mathfrak{S}_f$ is a quantum
  distinguishability quantifier that restricted to commuting pairs of quantum
  states, namely such that $[\rho, \sigma] = 0$, concides with the classical
  $f$-divergence of the probability distributions $\tmmathbf{p}_{\rho}$ and
  $\tmmathbf{p}_{\sigma}$ of their eigenvalues
  \begin{eqnarray}
    \mathfrak{D}_f (\rho , \sigma) & = & D_f (\tmmathbf{p}_{\rho},
    \tmmathbf{p}_{\sigma}) .  \label{eq:qfd}
  \end{eqnarray}
\end{definition}

In analogy to Sect.~\ref{sect:suff}  we can introduce the following
definition.

\begin{definition}
  (Invertibility for completely positive trace preserving maps)
  
  Consider a collection $\mathcal{P} $ of quantum states and a
  completely positive trace preserving map $\Phi$. 
If $\Phi$ is invertible on any  $\rho \in\mathcal{P} $, with the inverse
  $\tilde{\Phi}$ completely positive and trace preserving, the completely
  positive trace preserving map $\Phi$ is said to be invertible (or sufficient) with respect
  to the collection $\mathcal{P} $.
\end{definition}

Also in the quantum framework we can ask whether a distinguishability quantifier $\mathfrak{D}$ is such that validity of   
\begin{equation}   
    \mathfrak{D}( \Phi [\rho], \Phi [\sigma])=\mathfrak{D}( \rho, \sigma)  \label{eq:suffa}
\end{equation}
on a set of states implies invertibility of the map $\Phi$ on this set of states, i.e.:
\begin{equation}   
    \mathfrak{D}( \Phi [\rho], \Phi [\sigma])=\mathfrak{D}( \rho, \sigma) \Rightarrow \Phi\text{  sufficient on }\rho\text{ and }\sigma,  \label{eq:inv_suff_q}
\end{equation}
see Eq. \eqref{eq:inv_suff} for the classical analogue.
Importantly, this property was shown to hold for the natural quantum counterpart of the
Kullback-Leibler divergence, namely the Umegaki's quantum relative entropy by
Petz {\cite{Petz1986a}}, thus implying the same property for the Holevo
quantity {\cite{Shirokov2013a}}. We refer the interested reader to
\cite{Jencova2012a,Hiai2017a} and references therein for a proper account of
these features, as well as to {\cite{Jencova2006a,Buscemi2012a,Nagasawa2024a}} for the notion of sufficiency in the quantum framework. This richness and complexity of situations is a recurring
paradigm in moving from a classical to a quantum description. Operator
ordering in replacing classical variables, given for example by functions of
position and momentum, with quantum observables immediately appears as one of
the important topics in treatment of quantum mechanics inspired by classical
mechanics, see e.g. {\cite{Zachos2005a}}. This is all the more true when
investigating features of quantum mechanics starting from concepts and results
in quantum probability, a very natural perspective in view of the fact that
indeed quantum mechanics is a probability theory. In this scenario we can
recall e.g. quantum counterparts of Markovian
{\cite{Lindblad1976a,Gorini1976a,Vacchini2005a}} or non-Markovian
{\cite{Breuer2008a,Vacchini2016b,Vacchini2020a}} evolutions, where again the
non-commutativity of the relevant operator spaces leads to a variety of new
features.

Despite these facts, the unifying perspective suggested by classical
$f$-divergences still remains very useful. We will indeed consider various
distinguishability quantifiers used in the physical literature that fall within the class of quantum $f$-divergences for a suitable convex function $f$.

\subsection{Relevant examples of quantum 
\texorpdfstring{$f$-divergences}{f-divergences}}
\label{sect:fqex}

We thus introduce here some relevant examples of quantum $f$-divergences, providing
the natural counterparts of the quantities considered in Sect.~\ref{sect:fex},
that have recently been considered in order to estimate the non-Markovianity
of a quantum dynamical evolution {\cite{Breuer2009b,Megier2021a}}.

\subsubsection*{Quantum relative entropy}The most relevant, corresponding to
Eq.~(\ref{eq:fkl}), is the quantum relative entropy given by
\begin{eqnarray}
  S (\rho , \sigma) & = & \tmop{Tr} \{ \rho \log \rho \} - \tmop{Tr} \{ \rho
  \log \sigma \}  \label{eq:qrel}
\end{eqnarray}
when the support of the statistical operator $\sigma$ contains the support of
$\rho$, and infinite otherwise.

\subsubsection*{Quantum skew divergence}Recently, it has been put into
evidence that this quantity can be modified to obtain a distinguishability
quantifier that is bounded and can be further shown to obey suitable
triangle-like inequalities {\cite{Audenaert2014c}}, according to
\begin{eqnarray}
  S_{\mu} (\rho , \sigma) & = & \frac{\mu}{\log (1 / \mu)} S (\rho , \mu \rho
  + (1 - \mu) \sigma) \nonumber\\
  &  & + \frac{1 - \mu}{\log (1 / (1 - \mu))} S (\sigma , (1 - \mu) \sigma +
  \mu \rho)  \label{eq:QSD}
\end{eqnarray}
where $\mu$ is a skew or mixing parameter taking values in $(0, 1)$ and the
expression is symmetric under the exchange $\{ \mu, \rho \} \leftrightarrow \{
1 - \mu, \sigma \}$. A closely related quantity is the Holevo skew divergence,
which shares the same symmetry and also takes values in the finite interval
$[0, 1]$, namely
\begin{eqnarray}
  K_{\mu} (\rho , \sigma) & = & \frac{\mu}{h (\mu)} S (\rho , \mu \rho + (1 -
  \mu) \sigma) \nonumber\\
  &  & + \frac{1 - \mu}{h (1 - \mu)} S (\sigma , (1 - \mu) \sigma + \mu
  \rho),  \label{eq:HSD}
\end{eqnarray}
with $h (\mu) = - \mu \log \mu - (1 - \mu) \log (1 - \mu)$ the binary entropy
associated to the distribution $\{ \mu, 1 - \mu \}$. The latter quantity
corresponds to the Holevo quantity for the ensemble $\{ \mu, \rho ; 1 - \mu,
\sigma \}$, which suitably normalized has all the features of
distinguishability quantifiers and can be seen as an $f$-divergence for $f =
f^{\mu}_{\tmop{hsd}} (t) / h (\mu)$, that is we can write
\begin{eqnarray}
  K_{\mu} (\rho , \sigma) & = & \frac{1}{h (\mu)} \chi (\{ \mu, \rho ; 1 -
  \mu, \sigma \})  \label{eq:holhol}
\end{eqnarray}
where $\chi$ denotes the Holevo quantity for the considered ensemble of two
elements, according to the general definition {\cite{Nielsen2000}}
\begin{equation}
  \chi (\{ \mu_1, \rho_1 ; \ldots ; \mu_n, \rho_n \})  =  S \left( \sum_{i =
  1}^n \mu_i \rho_i \right) - \sum_{i = 1}^n \mu_i S (\rho_i) \bignone . 
\end{equation}
Indeed, Eq.~(\ref{eq:HSD}) can also be seen as a particular instance of Donald's
identity {\cite{Bengtsson2017}}. 
We stress that the property of 
the quantum relative entropy {\cite{Petz1986a}}, that entails sufficiency of a completely positive trace preserving map upon saturation of the data processing inequality, see Eq. \eqref{eq:inv_suff_q}, extends to these related
quantum distinguishability quantifiers {\cite{Shirokov2013a}}.

\subsubsection*{Trace distance}Also the $f$-divergence associated to the
convex function Eq.~(\ref{eq:fqsd}) has a quantum counterpart given by a
distance, known as trace distance
\begin{eqnarray}
  D (\rho , \sigma) & = & \frac{1}{2} \| \rho  - \sigma \|_1, 
  \label{eq:traced}
\end{eqnarray}
where $\| \cdot \|_1$ denotes the trace norm {\cite{Reed1980}}.

\subsubsection*{Quantum Jensen-Shannon divergence}As last example we recall
that for the case $\mu = \tfrac{1}{2}$ both Eq.~(\ref{eq:QSD}) and
(\ref{eq:HSD}) collapse to a unique quantum Jensen-Shannon divergence
{\cite{Majtey2005a,Lamberti2008a}}
\begin{eqnarray}
  \mathsf{J} (\rho , \sigma) & = & \frac{1}{2} \left[ S \left( \rho ,
  \frac{\rho + \sigma}{2} \right) + S \left( \sigma , \frac{\rho + \sigma}{2}
  \right) \right], 
\end{eqnarray}
and also in this case it has been shown that it has the property of a squared
distance, so that its square root is a proper distance between quantum states
{\cite{Sra2021a,Virosztek2021a}}.

\section{Optimal pairs}\label{sect:mp}

We now formulate and prove our main result. We characterize the pairs of
quantum states on which a bounded quantum distinguishability quantifier
contractive under the action of a completely positive trace preserving map
takes on its maximum value $M$.

\subsection{Orthogonality as a necessary and sufficient
condition}\label{sect:ortons}

The following lemma holds as shown in {\cite{Chefles2000a}}.

\begin{lemma}
  \label{lemma:lem}Orthogonality of a pair of pure states is a sufficient
  condition for the existence of a completely positive trace preserving map
  that transforms it in an arbitrary target pair of pure states. The condition
  becomes necessary if the target pair is also orthogonal.
\end{lemma}

We can now prove the main result of the paper.

\begin{theorem}
  \label{th:teo}Consider a quantum distinguishability quantifier
  $\mathfrak{S}$, that satisfies:
  \begin{eqnarray}
    \mathfrak{S} (\Phi [\rho], \Phi [\sigma]) & \leqslant & \mathfrak{S} (\rho
    , \sigma) \nonumber
  \end{eqnarray}
  for any pair of quantum states $\rho, \sigma \in \mathcal{S (\mathcal{H})}$
  and any completely positive trace preserving map $\Phi$,
  and such that
  \begin{eqnarray}
    \mathfrak{S} (\Phi [\rho], \Phi [\sigma]) & = & \mathfrak{S} (\rho ,
    \sigma) \nonumber
  \end{eqnarray}
  with $[\rho, \sigma] = 0$ implies the existence of a completely positive
  trace preserving map $\tilde{\Phi}$ fulfilling
  \begin{eqnarray}
    \tilde{\Phi} \circ \Phi [\rho] & = & \rho \nonumber
  \end{eqnarray}
  and
  \begin{eqnarray}
    \tilde{\Phi} \circ \Phi [\sigma] & = & \sigma,  \nonumber
  \end{eqnarray}
i.e.  the sufficiency of $\Phi$ on these states.
  Then
  \begin{eqnarray}
    \mathfrak{S} (\rho_1, \rho_2) & = & M, \nonumber
  \end{eqnarray}
  with $M$ the maximum value that $\mathfrak{S}$ can attain, if and only if
 $\tmop{supp} (\rho_1) \perp \tmop{supp} (\rho_2)$
so that in particular $[\rho_1,
  \rho_2] = 0$.
\end{theorem}

\begin{proof}
  We denote as $M$ the maximum value attained by the distinguishability
  quantifier $\mathfrak{S}$, with $M \in \{ 1, + \infty \}$, according to a
  proper choice of normalization if $\mathfrak{S}$ is bounded.
  
  We first show that the maximum is attained on pairs of orthogonal states.
  Indeed, given an arbitrary pair of states $\{ \sigma_1, \sigma_2 \}$, we can
  find a pair of states $\{ \rho_1, \rho_2 \}$, that are orthogonal, i.e. with
  orthogonal support, and such that
  \begin{eqnarray}
    \mathfrak{S} (\sigma_1, \sigma_2) & \leqslant & \mathfrak{S} (\rho_1,
    \rho_2) .  \label{eq:start1}
  \end{eqnarray}
  We first recall the fact shown in {\cite{Kleinmann2006a,Kleinmann2007a}}
  that the purification of a pair of orthogonal states can be obtained by
  means of a completely positive trace preserving map, that admits a
  completely positive trace preserving inverse on this pair of states, that
  is, according to Sect.~\ref{sect:suff}, is sufficient with respect to this
  pair. Denoting this map as
  \begin{eqnarray}
    \Lambda_{\tmop{pur}} : \mathcal{T} (\mathcal{H}_S \nobracket) &
    \rightarrow & \mathcal{T} (\mathcal{H}_S \otimes \mathcal{H}_A
    \nobracket), 
  \end{eqnarray}
  where we have denoted as $\mathcal{H}_S \otimes \mathcal{H}_A$ the Hilbert
  space in which the purifications have been obtained, thanks to the
  invariance property Eq.~(\ref{eq:sussq}) of Sect.~\ref{sect:qm} we have
  \begin{eqnarray}
    \mathfrak{S} (\Lambda_{\tmop{pur}} [\rho_1], \Lambda_{\tmop{pur}}
    [\rho_2]) & = & \mathfrak{S} (\rho_1, \rho_2) .  \label{eq:aux1}
  \end{eqnarray}
  We recall that the pure states $\Lambda_{\tmop{pur}} [\rho_1]$ and
  $\Lambda_{\tmop{pur}} [\rho_2]$ are still orthogonal, as directly follows
  considering their expressions in terms of vectors spanning the eigenspace of
  $\rho_1$ and $\rho_2$ respectively. We can now consider an arbitrary
  purification of the pair $\{ \sigma_1, \sigma_2 \}$, that we denote as $\{
  P_{\sigma_1}, P_{\sigma_2} \}$, with  \{$P_{\sigma_i}$\}$_{i = 1, 2}$ pure
  states in the Hilbert space $\mathcal{H}_S \otimes \mathcal{H}_A$, so that
  thanks to complete positivity and trace preservation of the partial trace we
  have
  \begin{equation}
    \mathfrak{S} (\sigma_1, \sigma_2) =\mathfrak{S} (\tmop{Tr}_A
    [P_{\sigma_1}], \tmop{Tr}_A [P_{\sigma_2}]) \leqslant \mathfrak{S}
    (P_{\sigma_1}, P_{\sigma_2}) . \label{eq:aux2}
  \end{equation}
  As a final step we exploit Lemma \ref{lemma:lem}. We can thus introduce a
  completely positive trace preserving map
  \begin{eqnarray}
    \Phi : \mathcal{T} (\mathcal{H}_S \otimes \mathcal{H}_A \nobracket) &
    \rightarrow & \mathcal{T} (\mathcal{H}_S \otimes \mathcal{H}_A \nobracket)
    \label{eq:a1}
  \end{eqnarray}
  such that
  \begin{eqnarray}
    \Phi [\Lambda_{\tmop{pur}} [\rho_i]] & = & P_{\sigma_i} \qquad i = 1, 2. 
    \label{eq:a2}
  \end{eqnarray}
  We thus have the inequality
  \begin{eqnarray}
    \mathfrak{S} (P_{\sigma_1}, P_{\sigma_2}) & \leqslant & \mathfrak{S}
    (\Lambda_{\tmop{pur}} [\rho_1], \Lambda_{\tmop{pur}} [\rho_2]), 
  \end{eqnarray}
  that combined with Eq.~(\ref{eq:aux2}) and Eq.~(\ref{eq:aux1}) proves
  Eq.~(\ref{eq:start1}).
  
  Moreover, the distinguishability quantifier $\mathfrak{S}$ takes on the same
  maximum value on each orthogonal pair of states. Consider two pairs of
  orthogonal states $\{ \rho_1, \rho_2 \}$ and $\{ \tilde{\rho}_1,
  \tilde{\rho}_2 \}$. According to the previous discussion we can introduce
  two completely positive trace preserving maps $\Lambda_{\tmop{pur}}$ and
  $\tilde{\Lambda}_{\tmop{pur}}$ sufficient with respect to the two distinct
  pairs so that both Eq.~(\ref{eq:aux1}) and
  \begin{eqnarray}
    \mathfrak{S} (\tilde{\Lambda}_{\tmop{pur}} [\tilde{\rho}_1],
    \tilde{\Lambda}_{\tmop{pur}} [\tilde{\rho}_2]) & = & \mathfrak{S}
    (\tilde{\rho}_1, \tilde{\rho}_2)  \label{eq:aux11}
  \end{eqnarray}
  hold. Furthermore thanks to orthogonality each purification according to
  Lemma \ref{lemma:lem} can be reached from the other by means of a completely
  positive trace preserving map as in Eq.~(\ref{eq:a2}) according to
  \begin{eqnarray}
    \Phi_{\rightarrow} [\Lambda_{\tmop{pur}} [\rho_i]] & = &
    \widetilde{\Lambda }_{\tmop{pur}} [\tilde{\rho}_i] \qquad i = 1, 2, 
    \label{eq:a11}
  \end{eqnarray}
  and
  \begin{eqnarray}
    \Phi_{\leftarrow} [\widetilde{\Lambda }_{\tmop{pur}} [\tilde{\rho}_i]] & =
    & \Lambda_{\tmop{pur}} [\rho_i] \qquad i = 1, 2,  \label{eq:a22}
  \end{eqnarray}
  respectively. Combining Eq.~(\ref{eq:aux11}) and Eq.~(\ref{eq:a11}) we
  obtain applying twice Eq.~(\ref{eq:inv1})
  \begin{eqnarray}
    \mathfrak{S} (\tilde{\rho}_1, \tilde{\rho}_2) & \leqslant & \mathfrak{S}
    (\rho_1, \rho_2),  \label{eq:restart1}
  \end{eqnarray}
  while combining Eq.~(\ref{eq:aux1}) and Eq.~(\ref{eq:a22}) we obtain
  \begin{eqnarray}
    \mathfrak{S} (\rho_1, \rho_2) & \leqslant & \mathfrak{S} (\tilde{\rho}_1,
    \tilde{\rho}_2),  \label{eq:restart2}
  \end{eqnarray}
  and therefore
  \begin{equation}
    \mathfrak{S} (\rho_1, \rho_2) =\mathfrak{S} (\tilde{\rho}_1,
    \tilde{\rho}_2) = M.
  \end{equation}

  We finally show that if on a pair of states the distinguishability
  quantifier $\mathfrak{S}$ attains its maximum value $M$, the pair has to be
  orthogonal. Indeed, supposing
  \begin{eqnarray}
    \mathfrak{S} (\sigma_1, \sigma_2) & = & M,  \label{eq:start2}
  \end{eqnarray}
  with $\sigma_1$ and $\sigma_2$ non orthogonal we come to a contradiction.
  Consider again a purification $\{ P_{\sigma_1}, P_{\sigma_2} \}$ on
  $\mathcal{H}_S \otimes \mathcal{H}_A$ of the pair $\{ \sigma_1, \sigma_2
  \}$. According to Lemma \ref{lemma:lem} these pure states can be obtained
  acting with a completely positive trace preserving map $\Phi$ as in
  Eq.~(\ref{eq:a1}) on a pair of pure and orthogonal states in $\mathcal{H}_S
  \otimes \mathcal{H}_A$, that we denote as $\{ P_{\rho_1}, P_{\rho_2} \}$. We
  thus have the chain of inequalities
  \begin{eqnarray}
    &&\mathfrak{S} (\sigma_1, \sigma_2) =\mathfrak{S} (\tmop{Tr}_A
    [P_{\sigma_1}], \tmop{Tr}_A [P_{\sigma_2}]) \label{eq:catena}\\
    &&\leqslant \mathfrak{S}
    (P_{\sigma_1}, P_{\sigma_2}) =\mathfrak{S} (\Phi [P_{\rho_1}], \Phi
    [P_{\rho_2}]) \leqslant \mathfrak{S} (P_{\rho_1}, P_{\rho_2}),
    \nonumber
  \end{eqnarray}
  but orthogonality of the pair $\{ P_{\rho_1}, P_{\rho_2} \}$ together with
  the hypothesis Eq.~(\ref{eq:start2}) implies that all terms in
  Eq.~(\ref{eq:catena}) are equal. We have in particular
  \begin{eqnarray}
    \mathfrak{S} (\Phi [P_{\rho_1}], \Phi [P_{\rho_2}]) & = & \mathfrak{S}
    (P_{\rho_1}, P_{\rho_2}).  \label{eq:start2bis}
  \end{eqnarray}
  
 Now we exploit the fact that
  invariance of $\mathfrak{S}$ under the action of $\Phi$ when restricted to a
  pair of commuting states, that is the equality sign in the data processing
  inequality, implies invertibility of the stochastic map obtained restricting
  $\Phi$ to such pairs, as discussed in Sect.~\ref{sect:suff}. Given this condition, there exists another stochastic map
  $\tilde{\Phi}$ acting as an inverse of $\Phi$ on this pair and therefore
  such that
  \begin{equation}
    P_{\rho_i} = \tilde{\Phi} [\Phi [P_{\rho_i}]] = \tilde{\Phi}
    [P_{\sigma_i}] \qquad i = 1, 2,
  \end{equation}
  which according to Lemma \ref{lemma:lem} implies orthogonality of the pair $\{ P_{\sigma_1}, P_{\sigma_2} \}$ and, in consequence, the orthogonality of the pair
  $\{ \sigma_1, \sigma_2 \}$, against the hypothesis.
\end{proof}

Note that orthogonality of the pair of quantum states in turn implies that the two states can be
  diagonalized together, and can be considered as a pair of classical states,
  that is two probability distributions. 
  With this, the property of invariance of $\mathfrak{S}$ under the action of $\Phi$ when restricted to a
  pair of classical states, warrants distinguishability quantifier is
  a quantum $f$-divergence according to Def.~\ref{def:qfd} with $f$ a strictly
  convex function.

\subsection{Known examples}

The characterization of optimal pairs for quantum distinguishability
quantifiers has been obtained for specific cases of interest along different
lines of proof, relying on the explicit expression of the quantifier, as we
now recall for relevant examples. The result in Sect.~\ref{sect:ortons}
provides a unified proof of general validity, pointing to contractivity under
completely positive trace preserving maps and the property that saturation of the data processing inequality of the corresponding
classical divergence implies sufficiency of the associated stochastic map as key ingredients.

For the case of the trace distance given by Eq.~(\ref{eq:traced}) the fact
that the trace distance attains its maximum value on pair of orthogonal states
directly follows from its definition, recalling that for a self-adjoint trace
class operator with spectral resolution
\begin{eqnarray}
  A & = & \sum_i a_i | \psi_i \rangle \langle \psi_i | \bignone 
  \label{eq:decompo}
\end{eqnarray}
we have
\begin{eqnarray}
  \| A \|_1 & = & \tmop{Tr} | A | \nonumber\\
  & = & \sum_i | a_i |, 
\end{eqnarray}
so that
\begin{eqnarray}
  \tmop{Tr} | \rho - \sigma | & = & 2 
\end{eqnarray}
iff $\tmop{supp} (\rho) \perp \tmop{supp} (\sigma)$. Also for the Bures
distance defined as {\cite{Bengtsson2017}}
\begin{eqnarray}
  D_B (\rho, \sigma) & = & \sqrt{1 - \tmop{Tr} \left| \sqrt{\sigma}
  \sqrt{\rho} \right|} 
\end{eqnarray}
and the Hellinger distance {\cite{Bengtsson2017}}
\begin{eqnarray}
  D_H (\rho, \sigma) & = & \sqrt{1 - \tmop{Tr} \left\{ \sqrt{\sigma}
  \sqrt{\rho} \right\}} 
\end{eqnarray}
the characterization of optimal pairs directly follows from the definition.

For the quantum skew divergence of Eq.~(\ref{eq:QSD}) a direct proof was given
in {\cite{Audenaert2014a}} building on operator monotonicity of the logarithm
for the values of the mixing parameter leading to a bounded quantifier, namely
$\mu \in (0, 1)$. This result further entails the same property for the
optimal pairs of the Holevo skew divergence of Eq.~(\ref{eq:HSD}).

\section{Counterexamples}\label{sect:cex}

We now provide simple examples to clarify the behaviour of distinguishability
quantifiers for which the contractivity property fails.

\subsection{Hilbert-Schmidt distance}

We consider the Hilbert-Schmidt distance between two states defined as
\begin{eqnarray}
  D_{\tmop{HS}} (\rho, \sigma) & = & \frac{1}{\sqrt{2}} \| \rho - \sigma \|_2,
  \label{eq:hsdist}
\end{eqnarray}
obtained from the Hilbert-Schmidt norm
\begin{eqnarray}
  \| A \|_2 & = & \sqrt{\tmop{Tr} \{ A^{\dag} A \}} 
\end{eqnarray}
of the difference of the two operators, so that it is jointly convex in its
arguments, with a prefactor chosen so that the distance takes values in the
range $[0, 1]$. Lack of contractivity of this distance under the action of a
completely positive trace preserving map was first put into evidence in
{\cite{Ozawa2000a}}.

\subsubsection*{Optimal pair}

We show that this distance takes on its maximum value only if the considered
states are both orthogonal and pure, so that orthogonality now becomes a
necessary but not sufficient condition for a pair to be optimal. Let us
consider an arbitrary pair of states $\rho, \sigma \in \mathcal{T
(\mathcal{H})}$ with orthogonal decompositions
\begin{eqnarray}
  \rho & = & \sum_i \rho_i | v_i \rangle \langle v_i | 
\end{eqnarray}
and
\begin{eqnarray}
  \sigma & = & \sum_k \sigma_k | w_k \rangle \langle w_k |, 
\end{eqnarray}
respectively. We then have, denoting with
\begin{eqnarray}
  \mathcal{P} (w) & = & \tmop{Tr} \{ w^2 \} 
\end{eqnarray}
the purity of an arbitrary state $w$, the identity
\begin{eqnarray}
  D_{\tmop{HS}}^2 (\rho, \sigma) & = & \frac{1}{2} \sum_{i, k} | \langle v_i |
  \nobracket w_k \rangle |^2 (\rho_i^2 - 2 \rho_i \sigma_k + \sigma_k^2)\\
  & = & \frac{1}{2} (\mathcal{P} (\rho) + \mathcal{P} (\sigma)) - \sum_{i, k}
  | \langle v_i | \nobracket w_k \rangle |^2 \rho_i \sigma_k . \nonumber
\end{eqnarray}
We thus obtain the inequality
\begin{eqnarray}
  D_{\tmop{HS}}^2 (\rho, \sigma) & \leqslant & \frac{1}{2} (\mathcal{P} (\rho)
  + \mathcal{P} (\sigma)),  \label{eq:puri}
\end{eqnarray}
that is saturated iff the states are orthogonal. The maximum value
$D_{\tmop{HS}}  (\rho, \sigma) = 1$ is attained if the pair is orthogonal and
both states have maximum purity, so that
\begin{eqnarray}
  \frac{1}{2} (\mathcal{P} (\rho) + \mathcal{P} (\sigma)) & = & 1. 
\end{eqnarray}
Orthogonality is therefore only a necessary condition for a pair to be among
the best distinguishable pairs according to the Hilbert-Schmidt distance. As an explicit example of states that are orthogonal but
not pure let us consider in $\mathbbm{C}^4$ the states
\begin{eqnarray}
  w_1 & = & \left(\begin{array}{cccc}
    \frac{1}{2} & 0 & 0 & 0\\
    0 & \frac{1}{2} & 0 & 0\\
    0 & 0 & 0 & 0\\
    0 & 0 & 0 & 0
  \end{array}\right) 
\end{eqnarray}
and
\begin{eqnarray}
  w_2 & = & \left(\begin{array}{cccc}
    0 & 0 & 0 & 0\\
    0 & 0 & 0 & 0\\
    0 & 0 & \frac{1}{2} & 0\\
    0 & 0 & 0 & \frac{1}{2}
  \end{array}\right) . 
\end{eqnarray}
We then have that $\tmop{supp} (w_1) \perp \tmop{supp} (w_2)$, but
$D_{\tmop{HS}}  (w_1, w_2) = 1 / \sqrt{2}$. Note that an example of states
that are orthogonal but not pure can only be obtained if the dimension of the
Hilbert space is higher than two. The same condition is necessary to show
non-contractivity of the Hilbert-Schmidt distance. This is simply due to the
fact that in $\mathbbm{C}^2$ the trace distance Eq.~(\ref{eq:traced}) and the
Hilbert-Schmidt distance Eq.~(\ref{eq:hsdist}) coincide up to a multiplicative
factor.

\subsubsection*{Lack of contractivity with respect to assignment
map\label{sect:loc}}

We note that while invariance under unitary transformations still hold, since
$D_{\tmop{HS}}  (\rho, \sigma)$ only depends on the spectrum of the operator
$\rho - \sigma$, invariance under the assignment map as in
Eq.~(\ref{eq:invTensor}) fails, despite the fact that Hilbert-Schmidt distance
is a contraction with respect to this completely positive trace preserving
map. To see this we consider Kadison's inequality
{\cite{Kadison1952a,Chruscinski2022a}}
\begin{eqnarray}
  (\Phi [A])^2 & \leqslant & \| \Phi \| \Phi [A^2], 
\end{eqnarray}
valid for a positive linear map $\Phi$ and a self-adjoint operator $A$, where
$\| \cdot \|$ denotes the uniform norm. Taking the trace of this expression
for a completely positive trace preserving map $\Phi$ we obtain
\begin{eqnarray}
  D_{\tmop{HS}}^2 (\Phi [\rho], \Phi [\sigma]) & \leqslant & \| \Phi
  [\mathbbm{1}] \|  D_{\tmop{HS}}^2 (\rho, \sigma),  \label{eq:ki}
\end{eqnarray}
where we have further used the fact that for a positive map $\| \Phi \|  = \|
\Phi [\mathbbm{1}] \| $ {\cite{Bhatia2007}}. If we now consider $\Phi$ to be
the assignment map we obtain
\begin{eqnarray}
  D_{\tmop{HS}}^2 (\rho \otimes \tau, \sigma \otimes \tau) & \leqslant & \|
  \mathbbm{} \tau \|  D_{\tmop{HS}}^2 (\rho, \sigma) .  \label{eq:bb1}
\end{eqnarray}
But for a statistical operator we have $\| \mathbbm{} \tau \| \leqslant 1$, so
that contractivity holds
\begin{eqnarray}
  D_{\tmop{HS}}  (\rho \otimes \tau, \sigma \otimes \tau) & \leqslant & 
  D_{\tmop{HS}}  (\rho, \sigma) . 
\end{eqnarray}
To obtain a finer estimate of the bound we consider an arbitrary statistical
operator $\tau$ on an auxiliary Hilbert space $\mathcal{K}$ with orthogonal
decomposition
\begin{eqnarray}
  \tau & = & \sum_i \tau_i | u_i \rangle \langle u_i |.  \label{eq:orto}
\end{eqnarray}
We then have the identity
\begin{eqnarray}
  D_{\tmop{HS}}^2 (\rho \otimes \tau, \sigma \otimes \tau) & = & \frac{1}{2}
  \tmop{Tr} ((\rho - \sigma) \otimes \tau)^2 \\
  & = & \frac{1}{2} \sum_{i, k, r} | \langle v_i | \nobracket w_k \rangle |^2
(\rho_i^2 - 2 \rho_i \sigma_k + \sigma_k^2) \tau_r^2,
\nonumber
\end{eqnarray}
and therefore
\begin{eqnarray}
  D_{\tmop{HS}} (\rho \otimes \tau, \sigma \otimes \tau) & = &
  \sqrt{\mathcal{P} (\tau)} D_{\tmop{HS}}  (\rho, \sigma),  \label{eq:assign}
\end{eqnarray}
leading to contractivity with respect to the assignment map
\begin{eqnarray}
  D_{\tmop{HS}} (\mathcal{A}_{\tau} [\rho], \mathcal{A}_{\tau} [\sigma]) &
  \leqslant & D_{\tmop{HS}}  (\rho, \sigma) .  \label{eq:assign2}
\end{eqnarray}
The exact estimate shows that invariance under the assignment map is therefore
generally obtained only when considering an auxiliary pure state. We recall in
passing that indeed for a statistical operator
\begin{equation}
  \mathcal{P} (\tau) = \| \tau \|_2^2 \leqslant \| \tau \| ,
\end{equation}
where in the notation of Eq.~(\ref{eq:orto}) we have
\begin{equation}
  \| \tau \|_2^2 = \sum_i \tau_i^2 \leqslant \tau_{\max} \sum_i \tau_i
  \bignone = \tau_{\max} = \| \tau \|,
\end{equation}
where $\tau_{\max}$ denoted the largest eigenvalue of $\tau$, so that indeed
the bound Eq.~(\ref{eq:bb1}) is consistent with Eq.~(\ref{eq:assign}).

\subsubsection*{Lack of contractivity with respect to partial
trace\label{sect:loc2}}

To provide an independent proof of failure of contractivity of the
Hilbert-Schmidt distance with respect to the one considered in
{\cite{Ozawa2000a}} we consider as completely positive trace preserving map
the partial trace. In order to exploit again Kadison's inequality
Eq.~(\ref{eq:ki}) for the case at hand we recall
\begin{eqnarray}
  \| \tmop{Tr}_E \{ \mathbbm{1}_2 \otimes \mathbbm{1}_n \} \| & = & n \|
  \mathbbm{1}_2 \| \nonumber\\
  & = & n, 
\end{eqnarray}
so that Eq.~(\ref{eq:ki}) leads to
\begin{eqnarray}
  D_{\tmop{HS}}  (\tmop{Tr}_E [\rho], \tmop{Tr}_E [\sigma]) & \leqslant &
  \sqrt{n}  D_{\tmop{HS}}  (\rho, \sigma),  \label{eq:strong}
\end{eqnarray}
pointing to the fact that failure of contractivity is now possible. Note that
this bound on the behaviour with respect to the partial trace can be obtained
along a different line of proof valid for an arbitrary Schatten norm
{\cite{Perez-Garcia2006a,Lidar2008a,Rastegin2012a}}.

Consider now in $\mathcal{H}_S \otimes \mathcal{H}_E$, with $\mathcal{H}_S
=\mathbb{C}^2$ and $\mathcal{H}_E =\mathbb{C}^n$, the states
\begin{eqnarray}
  \rho & = & P_+ \otimes \frac{1}{n} \mathbbm{1}_n  \label{eq:rho1}
\end{eqnarray}
and
\begin{eqnarray}
  \sigma & = & P_- \otimes \frac{1}{n} \mathbbm{1}_n,  \label{eq:rho2}
\end{eqnarray}
with $P_{\pm}$ projections on the eigenvectors of the Pauli matrix $\sigma_z$
with respect to the eigenvalues $+ 1$ and $- 1$ respectively. Acting with the
completely positive trace preserving map $\tmop{Tr}_E [\cdot]$ given by the
partial trace with respect to the second element of the tensor product we
obtain
\begin{eqnarray}
  \tmop{Tr}_E [\rho] & = & P_+ 
\end{eqnarray}
and
\begin{eqnarray}
  \tmop{Tr}_E [\sigma] & = & P_- 
\end{eqnarray}
respectively. We then have
\begin{eqnarray}
  D_{\tmop{HS}}  (\rho, \sigma) & = & \frac{1}{\sqrt{2}} \sqrt{\tmop{Tr}_{S
  \otimes E} \left\{ \!\left( P_+ \otimes \frac{1}{n} \mathbbm{1}_n - P_-
  \otimes \frac{1}{n} \mathbbm{1}_n \right)^2\! \right\}} \nonumber\\
  & = & \frac{1}{\sqrt{2 n^2}} \sqrt{\tmop{Tr}_{S \otimes E} \{ \mathbbm{1}_2
  \otimes \mathbbm{1}_n \}} \nonumber\\
  & = & \frac{1}{\sqrt{n}}, 
\end{eqnarray}
to be compared with
\begin{eqnarray}
  D_{\tmop{HS}}  (\tmop{Tr}_E [\rho], \tmop{Tr}_E [\sigma]) & = &
  \frac{1}{\sqrt{2}} \sqrt{\tmop{Tr}_S \{ (P_+ - P_-)^2 \}} \nonumber\\
  & = & \frac{1}{\sqrt{2}} \sqrt{\tmop{Tr}_S \{ \sigma_z^2 \}} \nonumber\\
  & = & 1. 
\end{eqnarray}
Note that according to Eq.~(\ref{eq:strong}) this is indeed the strongest
violation of trace contractivity that can be obtained, since the pair of
states considered in Eq.~(\ref{eq:rho1}) and (\ref{eq:rho2}) leads to saturate
the inequality.

\subsection{Distinguishability measure based on maximization over states}

\

As a further counterexample we consider another quantity that has recently
been suggested as a useful distinguishability quantifier to estimate the
quantumness of maps {\cite{Budini2024a}}, related to classicality studies of
dissipative dynamics {\cite{Gu2019a,Chen2019a,Budini2023b}}. We define the
distance between two states according to the expression
\begin{eqnarray}
  {D  }_{\infty} (\rho, \sigma) & = & \max_{w \in \mathcal{T (\mathcal{H})}, w
  \geqslant 0, \tmop{Tr} \{ w \} = 1} \tmop{Tr} | w (\rho - \sigma) | . 
  \label{eq:budistance}
\end{eqnarray}
This quantity has been introduced making reference to an alternative
expression for the trace distance Eq.~(\ref{eq:traced}), namely
\begin{eqnarray}
  D   (\rho, \sigma) & = & \max_{P \in \mathcal{B (\mathcal{H})}, P =
  P^{\dag}, P^2} \tmop{Tr} | P (\rho - \sigma) |,  \label{eq:budistanceNO}
\end{eqnarray}
where the maximum is taken over all projections, while in
Eq.~(\ref{eq:budistance}) the maximum is taken over states. As shown in
{\cite{Budini2024a}}, the maximum is obtained for the rank-one state given by
an eigenvector with highest eigenvalue of the self-adjoint operator $\rho -
\sigma$.

Given that for a trace class self-adjoint operator with spectral decomposition
Eq.~(\ref{eq:decompo}) we have
\begin{eqnarray}
  \| A \|  & = & \max_i | a_i |, 
\end{eqnarray}
we can write
\begin{eqnarray}
  {D  }_{\infty} (\rho, \sigma) & = & \| \rho - \sigma \|  . 
  \label{eq:exprinf}
\end{eqnarray}
We still have
\begin{equation}
  {0 \leqslant D  }_{\infty} (\rho, \sigma) \leqslant 1,
\end{equation}
but the maximum is attained exactly only when $\tmop{supp} (\rho) \perp
\tmop{supp} (\sigma)$ and at the same time at least one of the statistical
operators is pure. As for the Hilbert-Schmidt distance we thus have an
additional constraint related to purity of the states. Contractivity under
completely positive trace preserving maps is indeed lost, as can be seen
considering the action of the partial trace. As also discussed in
{\cite{Rastegin2012a}} we have
\begin{eqnarray}
  D_{\infty}  (\tmop{Tr}_E [\rho], \tmop{Tr}_E [\sigma]) & \leqslant & n 
  D_{\infty}  (\rho, \sigma),  \label{eq:strongbis}
\end{eqnarray}
where the same choice of pair of states as in Sect.~\ref{sect:loc} leads to
saturation of the inequality.

Note that we still have invariance under unitary transformations, while the
assignment map acts as a contraction
\begin{eqnarray}
  {D  }_{\infty} (\mathcal{A}_{\tau} [\rho], \mathcal{A}_{\tau} [\sigma]) & =
  & \| \rho - \sigma \|  \| \tau \|  \nonumber\\
  & \leqslant & {D  }_{\infty} (\rho, \sigma) . 
\end{eqnarray}
We finally stress that according to its explicit expression
Eq.~(\ref{eq:exprinf}) the quantifier ${D  }_{\infty}$ is jointly convex in
its argumens. Both these features are shared with the Hilbert-Schmidt
distance.

\begin{figure}[h]
  \resizebox{.9\columnwidth}{!}{\includegraphics{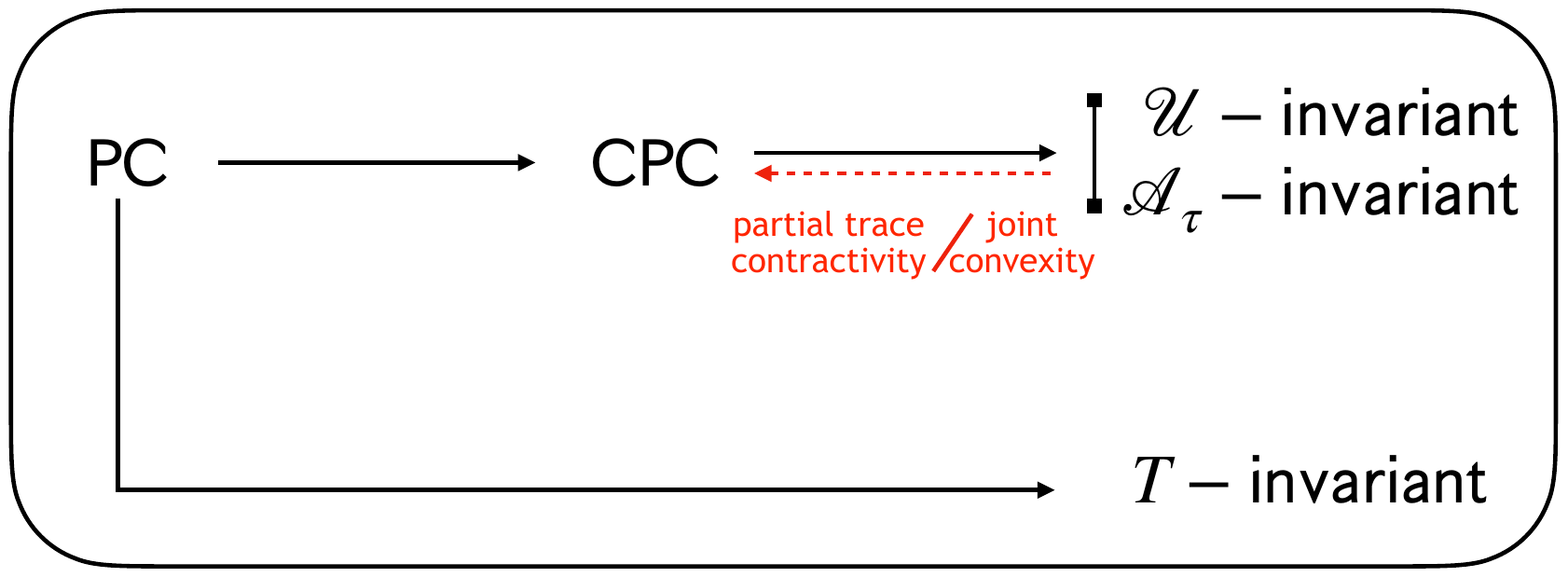}}
  \caption{\label{fig:schema}We graphically summarize the relationship between
  contractivity and invariance properties for a quantum distinguishability
  quantifier. The property of being contractive under positive trace
  preserving maps, which include completely positive trace preserving maps,
  entails invariance under transposition, unitary transformations and tensor
  product with a fixed state, as realized by the assignment map. If the
  quantum distinguishability quantifier is contractive under the partial trace
  operation or jointly convex in its arguments, invariance under unitary
  transformations and assignment maps in turn imply contractivity under any
  completely positive trace preserving map ($\mathsf{PC}$ = contracting under
  positive trace preserving maps; $\mathsf{CPC}$ = contracting under completely
  positive trace preserving maps, $T$-invariant = invariant under transposition map, $\mathcal{U}$-invariant = invariant under unitary transformations, $\mathcal{A}_\tau$-invariant = invariant under assignment map).}
\end{figure}

\section{Contractivity and invariances}\label{sect:cai}

In Sect.~\ref{sect:inv} and Sect.~\ref{sect:qm} we have considered certain
invariances of distinguishability quantifiers induced by the contractivity
property under the action of physical transformations. These invariances were
linked to the existence of a left inverse of the considered transformation. We
here use two alternative representations of completely positive maps to show
that under additional properties of the quantum distinguishability quantifier
these invariances are in fact equivalent to contractivity. As schematically
recalled in Fig.~\ref{fig:schema}, contractivity under the action of trace
preserving completely positive maps implies invariance under unitary
transformations as expressed by Eq.~(\ref{eq:invU}), as well as under the
action of assignment maps as expressed by Eq.~(\ref{eq:invTensor}). If the
considered distinguishability quantifier is also contractive under trace
preserving maps that are simply positive, then one also has invariance under
the transposition map, which is positive but not completely positive. This
happens both for the trace distance {\cite{Ruskai1994a}} and for the quantum
relative entropy {\cite{Reeb2017a}}, entailing this property for the quantum
skew divergence of Eq.~(\ref{eq:QSD}) and the Holevo skew divergence
Eq.~(\ref{eq:HSD}).

\subsubsection*{Contractivity under partial trace and equivalence between
contractivity and invariances}Given a completely positive trace preserving map
\begin{eqnarray}
  \Phi : \mathcal{T} (\mathcal{H}_S \nobracket) & \rightarrow & \mathcal{T}
  (\mathcal{H}_S \nobracket), 
\end{eqnarray}
it can always be represented in the form {\cite{Lindblad1975a}}
\begin{eqnarray}
  \Phi & = & \tmop{Tr}_E \circ \mathcal{U} \circ \mathcal{A}_{\tau} . 
  \label{eq:rapprLindblad}
\end{eqnarray}
It then immediately follows that given contractivity under partial trace,
invariance under unitary transformations and the assignment map imply
contractivity under any other completely positive transformation. We have in
fact
\begin{equation}
  \mathfrak{S} (\Phi [\rho], \Phi [\sigma]) = \mathfrak{S} (\tmop{Tr}_E
  \circ \mathcal{U} \circ \mathcal{A}_{\tau} [\rho] , \tmop{Tr}_E \circ
  \mathcal{U} \circ \mathcal{A}_{\tau} [\sigma])  \label{eq:sussqbis}
\end{equation}
and assuming contractivity under partial trace
\begin{eqnarray}
  \mathfrak{S} (\Phi [\rho], \Phi [\sigma]) & \leqslant & \mathfrak{S}
  (\mathcal{U} \circ \mathcal{A}_{\tau} [\rho] , \mathcal{U} \circ
  \mathcal{A}_{\tau} [\sigma]) 
\end{eqnarray}
so that the invariances warrant contractivity according to
Eq.~(\ref{eq:inv1}). This is exactly the path followed in the first proof of
contractivity of the quantum relative entropy {\cite{Lindblad1975a}}, the
contractivity with respect to the partial trace following from strong
subadditivity of entropy.

\subsubsection*{Joint convexity and equivalence between contractivity and
invariances} 
Restricting to the case in which the
representation Eq.~(\ref{eq:rapprLindblad}) holds with a finite-dimensional
environment, so that $\dim \mathcal{H}_E = d_E$, we now show that given joint
convexity of the considered distinguishability quantifier, as defined in
Eq.~(\ref{eq:jointc}), again invariance under unitary transformations and the
assignment map imply contractivity under any other completely positive
transformation. We first recall the relation {\cite{Collins2006a}}
\begin{eqnarray}
  \mathbb{E}_V [X] & = & \tmop{Tr} \{ X \} \frac{\mathbbm{1}}{n}, 
  \label{eq:haaraverage}
\end{eqnarray}
where the average is taken with respect to the Haar invariant measure over the
unitary group in dimension $n$, namely
\begin{eqnarray}
  \mathbb{E}_V [X] & = & \int_{V (n)} \mathd \mu_H (V) VXV^{\dag} . 
  \label{eq:2haar}
\end{eqnarray}
Combining Eq.~(\ref{eq:rapprLindblad}) and Eq.~(\ref{eq:haaraverage}) we then
have
\begin{eqnarray}
  \mathcal{A}_{\frac{\mathbbm{1 }_E}{d_E}} \circ \Phi & = & \mathbbm{1}_S
  \otimes \mathbb{E}_{V_E} \circ \mathcal{U} \circ \mathcal{A}_{\tau}, 
  \label{eq:rapprHaar}
\end{eqnarray}
where $\Phi$ is a completely positive trace preserving map acting on
$\mathcal{T} (\mathcal{H}_S \nobracket)$ and the map $\mathbbm{} \mathbbm{1
}_S \otimes \mathbb{E}_{V_E}$ according to Eq.~(\ref{eq:haaraverage}) and
Eq.~(\ref{eq:2haar}) acts on $X \in \mathcal{T} (\mathcal{H}_S \otimes
\mathcal{H}_E \nobracket)$ as
\begin{eqnarray}
  \mathbbm{1 }_S \otimes \mathbb{E}_{V_E} [X] & = & \mathbbm{1 }_S \otimes
  \int_{V_E (d_E)} \mathd \mu_H (V_E) V_E X  V_E^{\dag} \nonumber\\
  & = & \tmop{Tr}_E \{ X \} \otimes \frac{\mathbbm{1 }_E}{d_E} . 
\end{eqnarray}
For an arbitrary pair of states $\rho, \sigma \in \mathcal{T} (\mathcal{H}_S
\nobracket)$ we then have that assuming invariance under the assignment map
namely,
\begin{eqnarray}
  \mathfrak{S} (\Phi [\rho], \Phi [\sigma]) & = & \mathfrak{S} \left(
  \mathcal{A}_{\frac{\mathbbm{1 }_E}{d_E}} \circ \Phi [\rho] ,
  \mathcal{A}_{\frac{\mathbbm{1 }_E}{d_E}} \circ \Phi [\sigma] \right), 
  \label{eq:keyid}
\end{eqnarray}
thanks to Eq.~(\ref{eq:rapprHaar}) leads to
\begin{equation}
  \mathfrak{S} (\Phi [\rho], \Phi [\sigma])  =  \mathfrak{S} (\mathbbm{1}_S
  \otimes \mathbb{E}_{V_E} \circ \mathcal{U} \circ \mathcal{A}_{\tau} [\rho] ,
  \mathbbm{1}_S \otimes \mathbb{E}_{V_E} \circ \mathcal{U} \circ
  \mathcal{A}_{\tau} [\sigma]), 
\end{equation}
so that further considering joint convexity and normalization of the Haar
measure we finally obtain
\begin{eqnarray}
  \mathfrak{S} (\Phi [\rho], \Phi [\sigma]) & \leqslant & \int_{V_E (d_E)}
  \mathd \mu_H (V_E) \mathfrak{S} (\mathbbm{1 }_S \otimes V_E (\mathcal{U}
  \circ \mathcal{A}_{\tau} [\rho])  \nonumber\\
&&  \mathbbm{1}_S \otimes V_E^{\dag},\mathbbm{1 }_S \otimes V_E  (\mathcal{U} \circ \mathcal{A}_{\tau} [\sigma]) 
  \mathbbm{1}_S \otimes V_E^{\dag}) \nonumber\\
  & = & \int_{V_E (d_E)} \mathd \mu_H (V_E) \mathfrak{S} (\rho, \sigma)
  \nonumber\\
  & = & \mathfrak{S} (\rho, \sigma) . 
\end{eqnarray}
We note that joint convexity is a general property of classical divergences,
as shown at the end of Sect.~\ref{sect:inv}. In the quantum framework it holds
for the relative entropy, as well as the square of the Bures and Hellinger
distances that are known as Bures and Hellinger quantum divergences
respectively, while the property is not verified for the Bures and Hellinger
distance themselves {\cite{Spehner2017a}}.

Note that indeed both invariances, under unitary transformations and
assignment map, are needed for joint convexity to ensure contractivity under
an arbitrary completely positive trace preserving map, as shown in the
counterexamples considered in Sect.~\ref{sect:cex}, leading to failure of
Eq.~(\ref{eq:keyid}).

\section{Summary and conclusions}\label{sect:con}

In this work, we have established a direct connection between the
contractivity of quantum distinguishability quantifiers under completely
positive trace preserving maps and their optimal state pairs. 
Our main result demonstrates that for contractive quantifiers, such that saturation of the data processing inequality 
under the action of a completely positive trace preserving map on orthogonal states implies invertibility of the map on the given pair,
orthogonality is both a necessary and a sufficient condition for optimality. These distinguishability quantifiers include
the trace distance, the
quantum skew divergence and the Holevo skew divergence. In a complementary
way, we also examined the Hilbert-Schmidt distance and a distinguishability
quantifier based on a maximization over states for which contractivity under
completely positive trace preserving maps is lost, and we showed that for both
of them orthogonality remains a necessary condition for optimality but is no
longer sufficient, since also purity of the states is needed. Finally, we
have explored the interplay between contractivity and invariance properties,
showing that contractivity under completely positive trace preserving maps is
in one-to-one correspondence with invariance under unitary transformations and
assignment maps when additional conditions such as joint convexity or
contractivity under partial trace are satisfied.

Our results highlight the fundamental role of contractivity in determining the
behavior of quantum distinguishability measures, with potential applications
in quantum information theory, communication, and the study of open quantum
systems.

\begin{acknowledgments*}
  AS and BV acknowledge support from MUR and Next Generation EU via the PRIN
  2022 Project ``Quantum Reservoir Computing (QuReCo)'' (contract n.
  2022FEXLYB) and the NQSTI-Spoke1-BaC project QSynKrono (contract n.
  PE00000023-QuSynKrono).
\end{acknowledgments*}

\end{document}